\documentclass[a4paper]{amsart}
\usepackage{graphics}
\usepackage{vaucanson-g}
\usepackage{slashbox}
\usepackage{amssymb}

\theoremstyle{plain}
\newtheorem{theorem}{Theorem}
\newtheorem{lemma}[theorem]{Lemma}

\theoremstyle{definition}
\newtheorem{definition}[theorem]{Definition}
\newtheorem{remark}[theorem]{Remark}

\newtheorem{example}[theorem]{Example}

\DeclareMathOperator{\val}{val}
\DeclareMathOperator{\rep}{rep}
\DeclareMathOperator{\Fact}{Fact}
\DeclareMathOperator{\Shape}{Shape}

\newcommand{\kom}[1]{}

\title[Multidimensional Generalized Automatic Sequences\ldots]{Multidimensional Generalized Automatic Sequences and Shape-Symmetric Morphic Words}

\author[E. Charlier]{Emilie Charlier}
\author[T. K\"arki]{Tomi K\"arki${}^\star$}
\thanks{${}^\star$ Supported by Osk. Huttunen Foundation}
\author[M. Rigo]{Michel Rigo}
\address{Institute of Mathematics, University of Li\`ege,
Grande Traverse 12 (B 37), B-4000 Li\`ege, Belgium}
\email{\{echarlier,T.Karki,M.Rigo\}@ulg.ac.be}

\begin{document}

\begin{abstract}
    An infinite word is $S$-automatic if, for all $n\ge 0$, its $(n+1)$st
    letter is the output of a deterministic automaton fed with the
    representation of $n$ in the considered numeration system $S$. In
    this paper, we consider an analogous definition in a
    multidimensional setting and study the relationship with the
    shape-symmetric infinite words as introduced by Arnaud Maes.
    Precisely, for $d\ge 2$, we show that a multidimensional infinite
    word $x:\mathbb{N}^d\to\Sigma$ over a finite alphabet $\Sigma$ is
    $S$-automatic for some abstract numeration system $S$ built on a
    regular language containing the empty word if and only if $x$ is
    the image by a coding of a shape-symmetric infinite word.
\end{abstract}

\maketitle

\section{Introduction}

Let $k\ge2$. An infinite word $x=(x_n)_{n\ge 0}$ is {\em
  $k$-automatic} if for all $n\ge 0$, $x_n$ is obtained by feeding a
deterministic finite automaton with output (DFAO for short) with the
$k$-ary representation of $n$.  In his seminal paper \cite{Cob},
A.~Cobham shows that an infinite word is {\em $k$-automatic} if and
only if it is the image by a coding of a fixed point of a uniform
morphism of constant length $k$.

If we relax the assumption on the uniformity of the morphism, Cobham's
result still holds but $k$-ary systems are replaced by a wider class
of numeration systems, the so-called {\em abstract numeration systems}
\cite{LR,RM,R}. If an abstract numeration system is denoted by $S$,
the corresponding sequences that can be generated are said to be {\em
  $S$-automatic}. That is, the $(n+1)$st element of such a sequence is
obtained by feeding a DFAO with the representation of $n$ in the
considered abstract numeration system $S$.

This paper studies the relationship between sequences generated by
automata and sequences generated by morphisms, but extended to the
framework of multidimensional infinite words, i.e., maps from
$\mathbb{N}^d$ to some finite alphabet $\Sigma$. For instance,
$k$-automatic sequences have been generalized either by considering
$d$-tuples of $k$-ary representations given to a convenient DFAO or by
iterating morphisms for which images of letters are $d$-dimensional
cubes of constant size, see \cite{Salon} and also \cite{Per} for
questions related to frequencies of letters. In \cite{RM},
multidimensional $S$-automatic sequences have been introduced
mimicking O.~Salon's construction. Let us mention \cite{ABS} where a
different notion of bidimensional morphisms is introduced in
connection to problems arising in discrete geometry. In \cite{DFNR}
bidimensional $S$-automatic sequences turn out to be useful in the
context of combinatorial game theory. They play a central role to get
new caracterizations of $P$-positions for the famous Wythoff's game
and some of its variations. 
Another motivation for studying the set of
multidimensional $S$-automatic words $w$ over $\{0,1\}$ is to consider
them as characteristic words of subsets $P_w$ of $\mathbb{N}^d$, 
to extend the structure $\langle\mathbb{N};<\rangle$ 
by the corresponding predicates $P_w$ and to study the decidability 
of the corresponding first-order theory. See also \cite{CT} for 
relationship with second-order monadic theory.

Our main result in this paper can be precisely stated as follows.
\medskip

\noindent {\bf Theorem.}  Let $d\ge 1$. The $d$-dimensional infinite
word $x$ is $S$-automatic
for some abstract numeration system $S=(L,\Sigma,<)$ where
$\varepsilon\in L$ if and only if $x$ is the image by a coding
of a shape-symmetric $d$-dimensional infinite word.
\medskip

Our first task is to present the different concepts occurring in
this statement.  The notion of
shape-symmetry was first introduce by A.~Maes and was used mainly in
connection to logical questions about the decidability of first-order
theories where $\langle\mathbb{N};<\rangle$ is extended by some
morphic predicate \cite{Maes1,Maes2}.

\subsection{Abstract numeration systems} If $\Sigma$ is a finite
alphabet, $\Sigma^*$ denotes the free monoid generated by $\Sigma$
having concatenation of words as product and the empty word
$\varepsilon$ as neutral element. If $w=w_0\cdots w_{\ell-1}$ is a word, $\ell\ge0$, where $w_j$'s are letters, then $|w|$ denotes its
length $\ell$. Let $(\Sigma,<)$ be a totally ordered alphabet and $u,v$ be
two words over $\Sigma$. We say that $u$ is {\em genealogically less}
than $v$, and we write $u\prec v$ if either $|u|<|v|$ (i.e., $u$ is
of shorter length than $v$) or $|u|=|v|$ and there exist
$p,s,t\in\Sigma^*$, $a,b\in\Sigma$ such that $u=pas$, $v=pbt$ and
$a<b$ (i.e., $u$ is lexicographically less than $v$). Let us also
mention that we have taken the convention that all finite or infinite
words and pictures have indices starting from~$0$.

\begin{definition}
    An {\em abstract numeration system} \cite{LR} is a triple
    $S=(L,\Sigma,<)$ where $L$ is an infinite regular language over a
    totally ordered finite alphabet $(\Sigma,<)$. Enumerating the
    words of $L$ using the genealogical ordering $\prec$ induced by
    the ordering $<$ of $\Sigma$ gives a one-to-one correspondence
    $\rep_S:\mathbb{N}\to L$ mapping the non-negative integer $n$ onto
    the $(n+1)$st word in $L$.  In particular, $0$ is sent onto the
    first word in the genealogically ordered language $L$. The
    reciprocal map is denoted by $\val_S:L\to\mathbb{N}$.
\end{definition}

\begin{example}\label{exa:1}
    Take $\Sigma=\{a,b\}$ with $a<b$ and
    $L=\{a,ba\}^*\{\varepsilon,b\}$. The first words in $L$ are
    $\varepsilon$, $a$, $b$, $aa$, $ab$, $ba$, $aaa$, $aab$,\ldots.
    With $S=(L,\Sigma,<)$, we have for instance $\val_S(b)=2$ and
    $\rep_S(5)=ba$.
\end{example}

\begin{remark}
    Any positional numeration system built on a strictly increasing
    sequence $(U_n)_{n\ge 0}$ of integers such that $U_0=1$ gives an
    abstract numeration system whenever $\mathbb{N}$ is
    $U$-recognizable, i.e., whenever the set of greedy representations
    of the non-negative integers in terms of the sequence $(U_n)_{n\ge
      0}$ is regular.
\end{remark}

Any regular language is accepted by a
deterministic finite automaton, which is defined as follows. A
{\em deterministic finite automaton~$\mathcal{A}$} (DFA for short) is given by $\mathcal{A}=(Q,q_0,\Sigma,\delta,F)$ 
where $Q$ is the finite {\em set of states}, $q_0\in Q$ is the {\em initial state}, $\delta:Q\times\Sigma\to Q$ is the {\em transition function} 
and $F\subseteq Q$ is the {\em set of final
states}. The function $\delta$ can be extended to $Q\times\Sigma^*$
by $\delta(q,\varepsilon)=q$ for all $q\in Q$ and
$\delta(q,aw)=\delta(\delta(q,a),w)$ for all $q\in Q$, $a\in\Sigma$
and $w\in\Sigma^*$. A word $w\in\Sigma^*$ is {\em accepted by~$\mathcal{A}$} if $\delta(q_0,w)\in F$. 
The {\em langage accepted by~$\mathcal{A}$} is the set of the accepted words. 
A {\em deterministic finite automaton with output}
(DFAO for short) $\mathcal{B}=(Q,q_0,\Sigma,\delta,\Gamma,\tau)$  is
defined analogously where $\Gamma$ is the {\em output alphabet} and
$\tau:Q\to\Gamma$ is the {\em output function}. The {\em output corresponding
to the input $w\in\Sigma^*$} is $\tau(\delta(q_0,w))$.

%=============================================

\subsection{$S$-automatic multidimensional infinite words}

Let $d\ge 1$. To work with $d$-tuples of words of the same length,
we introduce the following map.

\begin{definition}\label{diese}
If $w_1,\ldots,w_d$ are finite words over the alphabet $\Sigma$, the
map $(\cdot)^{\#}:(\Sigma^*)^d\to{((\Sigma\cup\{\#\})^d)}^*$ is defined as
\begin{gather*}
(w_1,\ldots,w_d)^{\#}:=(\#^{m-|w_1|}w_1,\ldots,\#^{m-|w_d|}w_d)
\end{gather*}
where $m=\max\{|w_1|,\ldots,|w_d|\}$.
\end{definition}

As an example,
$(ab,bbaa)^{\#}=(\#\#ab,bbaa)$. In what follows,
we use the notation $\Sigma_{\#}$ as a shorthand for
$\Sigma\cup\{\#\}$.

\begin{definition}
    A {\em $d$-dimensional infinite word} over the
    alphabet $\Gamma$ is a map $x:\mathbb{N}^d\to\Gamma$. We use
    notation like $x_{n_1,\ldots,n_d}$ or $x(n_1,\ldots,n_d)$ to
    denote the value of $x$ at $(n_1,\ldots,n_d)$. Such a word is said
    to be {\em $S$-automatic} if there exist an abstract numeration
    system $S=(L,\Sigma,<)$ and a deterministic finite automaton with
    output $\mathcal{A}=(Q,q_0,(\Sigma_{\#})^d,\delta,\Gamma,\tau)$
    such that, for all $n_1,\ldots,n_d\ge 0$,
                $$\tau(\delta(q_0,(\rep_S(n_1),\ldots,\rep_S(n_d))^{\#}))=x_{n_1,\ldots,n_d}.$$
                This notion was introduced in \cite{RM} (see also \cite{NR}) as a
                natural generalization of the multidimensional $k$-automatic sequences
                introduced in \cite{Salon}.
\end{definition}

\begin{example} \label{exa:S-aut}
    Consider the abstract numeration system introduced in
    Example~\ref{exa:1}, $S=(\{a,ba\}^*\{\varepsilon,b\},\{a,b\},a<b)$
    and the DFAO depicted in Figure~\ref{fig:dfao1}.
    Since this automaton is fed with entries of the form
    $(\rep_S(n_1),\rep_S(n_2))^{\#}$,
    we do not consider the transitions of label $(\#,\#)$.
\begin{figure}[htbp]
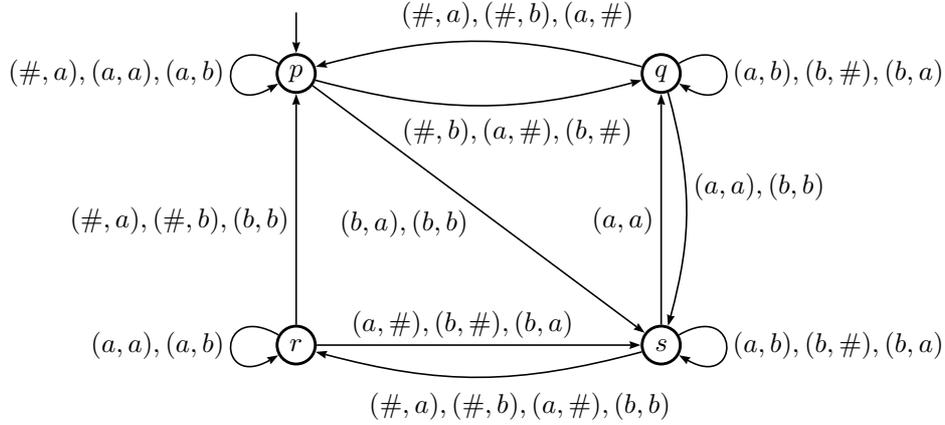

        \centering
\VCDraw{%
        \begin{VCPicture}{(0,-4.5)(8,4.5)}
% states
 \State[p]{(0,3)}{P}
 \State[q]{(8,3)}{Q}
 \State[r]{(0,-3)}{R}
 \State[s]{(8,-3)}{S}
% initial--final
\Initial[n]{P}
% transitions
\LoopW[.5]{P}{(\#,a),(a,a),(a,b)}
\ArcR[.6]{P}{Q}{(\#,b),(a,\#),(b,\#)}
\EdgeR{P}{S}{(b,a),(b,b)}
\EdgeL{R}{P}{(\#,a),(\#,b),(b,b)}
\EdgeL{R}{S}{(a,\#),(b,\#),(b,a)}
\LoopW[.5]{R}{(a,a),(a,b)}
\LoopE[.5]{S}{(a,b),(b,\#),(b,a)}
\ArcL{S}{R}{(\#,a),(\#,b),(a,\#),(b,b)}
\ArcL{Q}{S}{(a,a),(b,b)}
\LoopE[.5]{Q}{(a,b),(b,\#),(b,a)}
\ArcR{Q}{P}{(\#,a),(\#,b),(a,\#)}
\EdgeL{S}{Q}{(a,a)}
\end{VCPicture}
}
        \caption{A deterministic finite automaton with output.}
        \label{fig:dfao1}
    \end{figure}
    If the outputs of the DFAO are considered to be the states
    themselves, then we produce the bidimensional infinite $S$-automatic
    word given in Figure~\ref{tab:ex1}.
    \begin{figure}[htpb]
                    \begin{tabular}{r|ccccccccc}  & \rotatebox{90}{$\varepsilon$} & \rotatebox{90}{$a$} & \rotatebox{90}{$b$} & \rotatebox{90}{$aa$} & \rotatebox{90}{$ab$} & \rotatebox{90}{$ba$} & \rotatebox{90}{$aaa$}  & \rotatebox{90}{$aab$} & $\cdots$ \\
\hline
$\varepsilon$ & $p$ & $q$ & $q$ & $p$ & $q$ & $p$ & $q$ & $q$ & $\cdots$ \\
          $a$ & $p$ & $p$ & $s$ & $s$ & $q$ & $s$ & $p$ & $s$ &\\
          $b$ & $q$ & $p$ & $s$ & $q$ & $s$ & $q$ & $p$ & $s$ &\\
         $aa$ & $p$ & $p$ & $s$ & $p$ & $s$ & $q$ & $q$ & $s$ &\\
         $ab$ & $q$ & $p$ & $s$ & $p$ & $s$ & $s$ & $s$ & $r$ &\\
         $ba$ & $p$ & $s$ & $q$ & $p$ & $s$ & $q$ & $ s  $ & $q$ &\\
        $aaa$ & $p$ & $p$ & $s$ & $p$ & $s$ & $q$ & $ p  $ & $s$ &\\
        $aab$ & $q$ & $p$ & $s$ & $p$ & $s$ & $s$ & $ p  $ & $s$ &\\
$\vdots$ & $\vdots$ &  &  &  &  &  &   &  &$\ddots$\\
            \end{tabular}
\caption{A bidimensional infinite $S$-automatic word.}\label{tab:ex1}
    \end{figure}
\end{example}

%============================

\subsection{Multidimensional morphism} This section is given for the
sake of completeness and is mainly dedicated to present the notions of
multidimensional morphism and shape-symmetry as they were introduced
by A.~Maes mainly in connection with the
decidability question of logical theories \cite{Maes1,Maes2,Maes}.

If $i\le j$ are integers, $[\![i,j]\!]$ denotes the interval of
integers $\{i,i+1,\ldots,j\}$. Let $d\ge 1$.
If $\mathbf{n}\in\mathbb{N}^d$ and $i\in\{1,\ldots,d\}$, then
$n_i$ is the $i$th component of $\mathbf{n}$. Let
$\mathbf{m}$ and $\mathbf{n}$ be two
$d$-tuples in $\mathbb{N}^d$. We write $\mathbf{m}\le\mathbf{n}$ (resp. $\mathbf{m}<\mathbf{n}$), if
$m_i\le n_i$ (resp. $m_i<n_i$) for all $i=1,\ldots,d$. For
$\mathbf{n}\in\mathbb{N}^d$ and $j\in\mathbb{N}$,
$\mathbf{n}+j:=(n_1+j,\ldots,n_d+j)$. In particular, we set
$\mathbf{0}:=(0,\ldots,0)$ and $\mathbf{1}:=(1,\ldots,1)$.
If $j.\mathbf{1}\le\mathbf{n}$,
then we set $\mathbf{n}-j:=(n_1-j,\ldots,n_d-j)$.

\begin{definition}
    Let $s_1,\ldots,s_d$ be positive integers or $\infty$. A {\em
    $d$-dimensional picture} over the alphabet $\Sigma$ is a map $x$
    with domain $[\![0,s_1-1]\!]\times\cdots\times
    [\![0,s_d-1]\!]$ taking values in $\Sigma$. By convention, if
    $s_i=\infty$ for some $i$, then $[\![0,s_i-1]\!]=\mathbb{N}$.  If
    $x$ is such a picture, we write $|x|$ for the $d$-tuple
    $(s_1,\ldots,s_d)\in(\mathbb{N}\cup\{\infty\})^d$ which is called
    the {\em shape} of $x$. We denote by $\varepsilon_d$ the
    $d$-dimensional picture of shape $(0,\ldots,0).$
    Note that $\varepsilon_1=\varepsilon$.
    If $\mathbf{n}=(n_1,\ldots,n_d)$ belongs
    to the domain of $x$, we indifferently use the notation $x_{n_1,\ldots,n_d}$,
    $x_\mathbf{n}$, $x(n_1,\ldots,n_d)$ or $x(\mathbf{n})$.
    Let $x$ be a $d$-dimensional picture. If for all
    $i\in\{1,\ldots,d\}$, $|x|_i<\infty$, then $x$ is said to be {\em
    bounded}.  The set of $d$-dimensional bounded pictures over
    $\Sigma$ is denoted by $B_d(\Sigma)$. A bounded picture $x$ is a
    {\em square} of {\em size} $c\in\mathbb{N}$ if $|x|=c.\mathbf{1}$.
\end{definition}

\begin{definition}
    Let $x$ be a $d$-dimensional picture. If $\mathbf{0}\le
    \mathbf{s}\le\mathbf{t}\le |x|-1$, then $x[\mathbf{s},\mathbf{t}]$ is
    said to be a {\em factor} of $x$ and is defined as the picture $y$
    of shape $\mathbf{t}-\mathbf{s}+1$ given by
                $y(\mathbf{n})=x(\mathbf{n}+\mathbf{s})$ for all
    $\mathbf{n}\in\mathbb{N}^d$ such that $\mathbf{n}\le
    \mathbf{t}-\mathbf{s}$.
                For any $\mathbf{u}\in\mathbb{N}^d$,
                the set of factors of $x$ of shape $\mathbf{u}$ is denoted by
                $\Fact_{\mathbf{u}}(x)$.
\end{definition}

\begin{example}
    Consider the bidimensional (bounded) picture of shape $(5,2)$,
$$x=\begin{array}{|c|c|c|c|c|}
\hline
    a&b&a&a&b\\
\hline
    c&d&b&c&d\\
\hline
\end{array}\ .$$
We have
$$x[(0,0),(1,1)]=\begin{array}{|c|c|}
\hline
    a&b\\
\hline
    c&d\\
\hline
\end{array}\quad \text{ and }\quad
x[(2,0),(4,1)]=\begin{array}{|c|c|c|}
\hline
    a&a&b\\
\hline
    b&c&d\\
\hline
\end{array}\ .$$
For instance, $\Fact_{\mathbf{1}}(x)=\{a,b,c,d\}$ and
$$\Fact_{(3,2)}(x)=\left\{
\begin{array}{|c|c|c|}
\hline
    a&b&a\\
\hline
    c&d&b\\
\hline
\end{array}\, ,\
\begin{array}{|c|c|c|}
\hline
    b&a&a\\
\hline
    d&b&c\\
\hline
\end{array}\, ,\ \begin{array}{|c|c|c|}
\hline
    a&a&b\\
\hline
    b&c&d\\
\hline
\end{array}
\right\}.$$
\end{example}

\begin{definition}
                Let $x$ be a $d$-dimensional picture of shape
                $\mathbf{s}=(s_1,\ldots,s_d)$. For all $i\in\{1,\ldots,d\}$ and
                $k<s_i$, $x_{|i,k}$ is the $(d-1)$-dimensional picture of shape
                $$|x|_{\widehat{i}}=\mathbf{s}_{\widehat{i}}
                        :=(s_1,\ldots,s_{i-1},s_{i+1},\ldots,s_d)$$
                defined by
                $$x_{|i,k}(n_1,\ldots,n_{i-1},n_{i+1},\ldots,n_d)=
                                x(n_1,\ldots,n_{i-1},k,n_{j+1},\ldots,n_d)$$
                for all $0\le n_j < s_j$, $j\in\{1,\ldots,d\}\setminus\{i\}$.
\end{definition}

\begin{definition}
    Let $x,y$ be two $d$-dimensional pictures. If for some
    $i\in\{1,\ldots,d\}$,
    $|x|_{\widehat{i}}=|y|_{\widehat{i}}=(s_1,\ldots,s_{j-1},s_{j+1},\ldots,s_d)$,
    then we define the {\em concatenation} of $x$ and $y$ {\em in the
    direction $i$} as the $d$-dimensional picture
    $x\odot^i y$ of shape
    $(s_1,\ldots,s_{j-1},|x|_i+|y|_i,s_{j+1},\ldots,s_d)$ satisfying
    \begin{itemize}
      \item[(i)] $x=(x\odot^i y)[\mathbf{0},|x|-1]$
      \item[(ii)] $y=(x\odot^i y)
                                [(0,\ldots,0,|x|_i,0,\ldots,0),(0,\ldots,0,|x|_i,0,\ldots,0)+|y|-1]$.
    \end{itemize}
                The $d$-dimensional empty word $\varepsilon_d$ is a word of shape
                $\mathbf{0}$. We extend the definition to the concatenation of
                $\varepsilon_d$ and any $d$-dimensional word
                $x$ in the direction $i\in \{1,\ldots,d\}$ by
                \[
                \varepsilon_d \odot^i x = x \odot^i \varepsilon_d = x.
                \]
                Especially, $\varepsilon_d \odot^i \varepsilon_d = \varepsilon_d$.
\end{definition}

\begin{example}
    Consider the two bidimensional pictures
                $$x=
                \begin{array}{|c|c|}
                \hline
    a&b\\
                \hline
    c&d\\
                \hline
                \end{array}\quad \text{ and }\quad
                y=
                \begin{array}{|c|c|c|}
                \hline
    a&a&b\\
                \hline
    b&c&d\\
                \hline
                \end{array}$$
                of shape respectively $|x|=(2,2)$ and $|y|=(3,2)$. Since
                $|x|_{\widehat{1}}=|y|_{\widehat{1}}=2$, we get
                $$x\odot^1 y=\begin{array}{|c|c|c|c|c|}
                \hline
    a&b&a&a&b\\
                \hline
    c&d&b&c&d\\
                \hline
                \end{array}\ .$$
                But notice that $x\odot^2 y$ is not defined because
                $2=|x|_{\widehat{2}}\neq |y|_{\widehat{2}}=3$.
\end{example}

Let us now define how to erase hyperplanes from a multidimensional picture.

\begin{definition}\label{def:rho}
                Let $x$ be a $d$-dimensional picture of shape
                $(s_1,s_2,\ldots,s_d)$ over $\Sigma \cup \{e\}$,
                where $e$ does not belong to $\Sigma$.
                A $(d-1)$-dimensional picture $x_{|i,k}$ is called
                an \emph{$e$-hyperplane} of $x$ if each letter in $x_{|i,k}$
                is equal to~$e$. \emph{Erasing} an $e$-hyperplane $x_{|i,k}$
                of $x$ means replacing $x$ with a $d$-dimensional picture $x'= y \odot^i z$, where
                \[
                y= \left \{ \begin{array}{ll}
                x[\mathbf{0},(s_1,\ldots,s_{i-1},k,s_{i+1},\ldots,s_d)-1]
                & \text{if $k\ge1$,}\\
                \varepsilon_d & \text{otherwise,}
                \end{array}
                \right .
                \]
                and
                \[
                z= \left \{ \begin{array}{ll}
                x[(0,\ldots,0,k+1,0,\ldots,0),|x|-1)] & \text{if $k<s_i-1$},\\
                \varepsilon_d & \text{otherwise.}
                \end{array}
                \right .
                \]
                We denote by $\rho_e$ the map which associates to any $d$-dimensional
                picture $x$ over $\Sigma \cup \{e\}$, the picture
                $\rho_e(x)$ obtained by erasing iteratively every $e$-hyperplane of~$x$.
                Moreover, we say that $x$ is
                \emph{$e$-erasable} if the picture $\rho_e(x)$ does not contain
                the letter~$e$ as a factor anymore.
                In other words, for each position $\mathbf{n}$ such that
                $x_\mathbf{n}=e$, there exists an integer $i\in\{1,\ldots,d\}$
                such that $x_{|i,n_i}$ is an $e$-hyperplane.
\end{definition}

                Let $x$ be a $d$-dimensional picture and $\mu:\Sigma\to B_d(\Sigma)$
                be a map. Note that $\mu$ cannot necessarily be extended
                to a morphism on $\Sigma^*$. Indeed, if $x$ is a picture over $\Sigma$,
                $\mu(x)$ is not always well defined. Depending on the shapes
                of the images by $\mu$ of the letters in $\Sigma$, when trying to build
                $\mu(x)$ by concatenating the images $\mu(x_{\mathbf{i}})$
                we can obtain ``holes''
                or ``overlaps''. Therefore, we introduce some restrictions on $\mu$.

\begin{definition}
                Let $\mu:\Sigma\to B_d(\Sigma)$ be a map and $x$ be a $d$-dimensional
                picture such that
                $$\forall i\in\{1,\ldots,d\}, \forall k<|x|_i,
                \forall a,b\in \Fact_{\mathbf{1}}(x_{|i,k}) : |\mu(a)|_i=|\mu(b)|_i.$$
                Then $\mu(x)$ is defined as
                $$\mu(x)=\odot^1_{0\le n_1<|x|_1}\left(\cdots\left(\odot^d_{0\le n_d<|x|_d}                                     \mu(x(n_1,\ldots,n_d))\right)\right).$$
                Note that the ordering of the products in the different directions is unimportant.
\end{definition}

\begin{example}
    Consider the map $\mu$ given by
                $$a \mapsto
                \begin{array}{|c|c|}
                \hline
                a&a\\
                \hline
                b&d\\
                \hline
                \end{array}\, ,\ b\mapsto
                \begin{array}{|c|}
                \hline
                c\\
                \hline
                b\\
                \hline
                \end{array}\, ,\ c\mapsto
                \begin{array}{|c|c|}
                \hline
                a&a\\
                \hline
                \end{array}\, ,\ d\mapsto
                \begin{array}{|c|}
                \hline
                d\\
                \hline
                \end{array}\ .$$
                Let
                $$x=
                \begin{array}{|c|c|}
                \hline
    a&b\\
                \hline
    c&d\\
                \hline
                \end{array}\ .$$
                Since $|\mu(a)|_2=|\mu(b)|_2=2$,
                $|\mu(c)|_2=|\mu(d)|_2=1$, $|\mu(a)|_1=|\mu(c)|_1=2$ and
                $|\mu(b)|_1=|\mu(d)|_1=1$, $\mu(x)$ is well defined and given by
                $$\mu(x)=\begin{array}{|c|c|c|}
                \hline
                a&a&c\\
                \hline
                b&d&b\\
                \hline
                a&a&d\\
                \hline
                \end{array}\ .$$
                But one can notice that $\mu^2(x)$ is not well defined.
\end{example}

\begin{definition}
    Let $\mu:\Sigma\to B_d(\Sigma)$ be a map. If for all $a\in\Sigma$
    and all $n\ge 0$, $\mu^n(a)$ is well defined from $\mu^{n-1}(a)$,
    then $\mu$ is said to be a {\em $d$-dimensional morphism}.
\end{definition}

The usual notion of a {\em prolongable morphism} can be given in
this multidimensional setting.

\begin{definition}
                Let $\mu$ be a $d$-dimensional
    morphism and $a$ be a letter such that
    $(\mu(a))_{\mathbf{0}}=a$. We say that $\mu$ is {\em prolongable on $a$}.
    Then the limit
                $$w=\mu^\omega(a):=\lim_{n\to+\infty}\mu^n(a)$$
                is well defined and $w=\mu(w)$ is a {\em fixed point} of $\mu$.
                A $d$-dimensional infinite word $x$
                over $\Sigma$ is said to be {\em purely morphic} if it is a
                fixed point of a $d$-dimensional morphism.
                It is said to be {\em morphic} if there exists a coding $\nu:\Gamma\to\Sigma$                           (i.e., a letter-to-letter morphism) such that $x=\nu(y)$
                for some purely morphic word $y$ over~$\Gamma$.
\end{definition}

The so-called property of shape-symmetry that we introduce now is a
natural generalization of uniform morphisms where all images are
squares of the same dimension \cite{Salon}.

\begin{definition}
    Let $\mu:\Sigma\to B_d(\Sigma)$ be a $d$-dimensional morphism
    having the $d$-dimensional infinite word $x$ as a fixed point. If
    for any permutation $f$ of $\{1,\ldots,d\}$ and for all
    $n_1,\ldots,n_d>0$, $|\mu(x(n_{f(1)},\ldots,n_{f(d)}))|=
    (s_{f(1)},\ldots,s_{f(d)})$ whenever
    $|\mu(x(n_1,\ldots,n_d))|=(s_1,\ldots,s_d)$, then $x$ is said to
    be {\em shape-symmetric} (with respect to $\mu$).
\end{definition}

\begin{remark}
    An equivalent formulation of shape-symmetry is given as follows.
    Let $\mu:\Sigma\to B_d(\Sigma)$ be a $d$-dimensional morphism
    having the $d$-dimensional infinite word $x$ as a fixed point. This
    word is shape-symmetric if and only if
$$\forall i,j\le d,\forall k\in\mathbb{N},
\forall a\in \Fact_{\mathbf{1}}(x_{|i,k}), \forall b\in
    \Fact_{\mathbf{1}}(x_{|j,k}): |\mu(a)|_i=|\mu(b)|_j.$$
\end{remark}

\begin{remark}
    A. Maes showed that determining whether or not a map $\mu:\Sigma\to
    B_d(\Sigma)$ is a $d$-dimensional morphism is a decidable
    problem. Moreover he showed that if $\mu$ is prolongable on a letter $a$,
    then it is decidable whether or not the fixed
    point $\mu^\omega(a)$ is shape-symmetric \cite{Maes1,Maes2,Maes}.
\end{remark}

\begin{example}\label{exa:run}
    One can show that the following morphism has a fixed point
    $\mu^\omega(a)$ which is shape-symmetric.
$$\mu(a)=\mu(f)=
\begin{array}{|c|c|}
    \hline
a&b\\
\hline
c&d\\
\hline
\end{array}\, ,\ \mu(b)=\begin{array}{|c|}
    \hline
e\\
\hline
c\\
\hline
\end{array}\, ,\ \mu(c)=\begin{array}{|c|c|}
    \hline
e&b\\
\hline
\end{array}\, ,\mu(d)=\begin{array}{|c|}
    \hline
f\\
\hline
\end{array}\, ,\ \mu(e)=
\begin{array}{|c|c|}
    \hline
e&b\\
\hline
g&d\\
\hline
\end{array}\, ,$$
$$\mu(g)=
 \begin{array}{|c|c|}
    \hline
h&b\\
\hline
\end{array}\, ,\
\mu(h)=
\begin{array}{|c|c|}
    \hline
h&b\\
\hline
c&d\\
\hline
\end{array}\ .$$
We have represented in Figure~\ref{fig:mu} the beginning of the picture. Some elements are underlined for the use of Example~\ref{exa:later}.
\begin{figure}
    $$\begin{array}{|cc|c|cc|cc|c|cc|c}
\hline
\underline{a}&\underline{b}&e&e&b&e&b&e&e&b&\cdots\\
c&d&\underline{c}&g&d&g&d&c&g&d& \\
\hline
e&b&f&e&\underline{b}&h&b&f&h&b& \\
\hline
e&b&e&a&b&e&b&e&h&b& \\
g&d&c&c&d&g&d&\underline{c}&c&d& \\
\hline
e&b&e&e&b&a&b&e&e&b& \\
g&d&c&g&d&c&d&c&g&d& \\
\hline
h&b&f&e&b&e&b&f&h&b& \\
\hline
e&b&e&e&b&e&b&e&a&b& \\
g&d&c&g&d&g&d&c&c&d& \\
\hline
\vdots&&&&&&&&&&\ddots\\
\end{array}$$
    \caption{A fixed point of $\mu$.}
    \label{fig:mu}
\end{figure}
\end{example}

\begin{definition}
    Let $d\ge 2$ and let $\mu:\Sigma\to B_d(\Sigma)$ be a $d$-dimensional morphism
    having the $d$-dimensional infinite word $x$ as a fixed point. The
    {\em shape sequence} of $x$ with respect to $\mu$ {\em in the direction
    $i\in\{1,\ldots,d\}$} is the sequence
                $$\Shape_{\mu,i}(x)=(|\mu(x_{|i,k})|_i)_{k\ge 0}.$$
    For a unidimensional morphism $\mu$ having the infinite
    word $x=x_0x_1x_2\cdots$ as a fixed point, the {\em shape sequence} of $x$ with
    respect to $\mu$ is $\Shape_\mu(x)=(|\mu(x_k)|)_{k\ge 0}$.
\end{definition}

\begin{remark}\label{rem:shape}
    Let $\mu:\Sigma\to B_d(\Sigma)$ be a $d$-dimensional morphism
    having the $d$-dimensional infinite word $x$ as a fixed point.
    Note that $x$ is shape-symmetric if and only if
$$\Shape_{\mu,1}(x)=\cdots=\Shape_{\mu,d}(x).$$
\end{remark}

%======================================
\section{Main result}

Let us recall that our goal is to prove the following result.

\begin{theorem}
    Let $d\ge 1$. The $d$-dimensional infinite word
    $x$ is $S$-automatic
    for some abstract numeration system $S=(L,\Sigma,<)$ where
    $\varepsilon\in L$ if and only if $x$ is the image by a coding
    of a shape-symmetric infinite $d$-dimensional word.
\end{theorem}

The case $d=1$ is proved in \cite{RM}. It is a natural
generalization of the classical Cobham's theorem from 1972 \cite{Cob}.
For the sake of clarity, we make the proof in the case $d=2$.
We split the proof into two parts.

\medskip

{\bf Part 1.}
Assume that $x=\nu(\mu^\omega(a))$ where
$\nu:\Sigma\to\Gamma$ is a coding and $\mu:\Sigma\to B_2(\Sigma)$ is a
$2$-dimensional morphism prolongable on $a$ such that
$y=\mu^\omega(a)$ is shape-symmetric. We show in this part that
$x$ is $S$-automatic for some $S=(L,\Sigma,<)$ where $\varepsilon\in L$.

\medskip

Let $Y_1=(y_{n,0})_{n\ge 0}$ be the first line of $y$. This word $Y_1$
is a unidimensional infinite word over a subset $\Sigma_1$ of
$\Sigma$.  It is clear that $Y_1$ is generated by a unidimensional morphism $\mu_1$
derived from $\mu$ (one has only to consider the first line occurring
in the images by $\mu$ of the letters in $\Sigma$).

\begin{definition}\label{def:dir}
    With each (unidimensional) morphism $\mu:\Sigma\to \Sigma^*$
    and with each letter $a\in\Sigma$ we can canonically associate
    a DFA denoted by $\mathcal{A}_{\mu,a}$ and defined as follows.
    Let $r_\mu:=\max_{b\in \Sigma}|\mu(b)|$. The alphabet of
    $\mathcal{A}_{\mu,a}$ is $\{0,\ldots,r_\mu-1\}$. The set of
    states is $\Sigma$. The initial state is $a$ and every state is final.
    The (partial)
    transition function $\delta_\mu$ is defined by
    $\delta_\mu(b,i)=(\mu(b))_i$, for all $b\in \Sigma$ and
    $i\in\{0,\ldots,|\mu(b)|-1\}$. The language
    accepted by $\mathcal{A}_{\mu,a}$ from which are removed the
    words having $0$ as a prefix is called the {\it directive language}
    of $(\mu,a)$ and is denoted by $L_{\mu,a}$. Note that
    $L_{\mu,a}$ is a prefix language since all states in
    $\mathcal{A}_{\mu,a}$ are final. In particular, we have
    $\varepsilon\in L_{\mu,a}$. The reason why we call it {\em
    directive} will be clear, see Lemma~\ref{lem:dir} and
    Lemma~\ref{cor:dir}.
\end{definition}

\begin{example}\label{exa:run2}
    Considering Example~\ref{exa:run}, $\Sigma_1=\{a,b,e\}$,
    $\mu_1:a\mapsto ab,\ b\mapsto e,\ e\mapsto eb$ and
    $Y_1=abeebebeebeebebeebebeeb\cdots$. The DFA associated with $(\mu_1,a)$ is
    depicted in Figure~\ref{fig:dfacar}.
\begin{figure}[htbp]
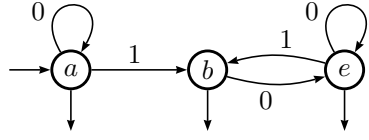

        \centering
\VCDraw{%
        \begin{VCPicture}{(0,-1)(6,1)}
% states
 \State[a]{(0,0)}{a}
 \State[b]{(3,0)}{b}
 \State[e]{(6,0)}{e}
% initial--final
\Initial{a}
\Final[s]{a}
\Final[s]{b}
\Final[s]{e}
% transitions
\LoopN{a}{0}
\EdgeL{a}{b}{1}
\ArcR{b}{e}{0}
\LoopN{e}{0}
\ArcR{e}{b}{1}
\end{VCPicture}
}
        \caption{The automaton $\mathcal{A}_{\mu_1,a}$.}
        \label{fig:dfacar}
    \end{figure}
The first words in the directive language of $(\mu_1,a)$ are
$$L_{\mu_1,a}=\{\varepsilon,1,10,100,101,1000,1001,1010,10000,\ldots\}.$$
\end{example}

\begin{lemma}\label{lem:dir}
    Let $\mu:\Sigma\to\Sigma^*$ be a morphism prolongable on
    $a\in\Sigma$. Let $S$ be the abstract numeration
    system built on the directive language $L_{\mu,a}$ of
    $(\mu,a)$ with the ordered alphabet $(\{0,\ldots,r_\mu-1\},0<\cdots<r_\mu-1)$.
    Then, for the infinite word $\mu^\omega(a)=y_0y_1y_2\cdots$ and for all $n\ge 0$,
    we have
    $$y_n=\delta_\mu(a,\rep_S(n))$$
    and
    $$\mu(y_n)=\mu^\omega(a)[\val_S(\rep_S(n)0),\val_S(\rep_S(n)(|\mu(y_n)|-1))].$$
\end{lemma}

\begin{proof}
    The adjacency matrix $M\in\mathbb{N}^{\Sigma\times\Sigma}$ of
    $\mathcal{A}_{\mu,a}$ is defined for all $b,c\in\Sigma$ by
    $M_{b,c}=\#\{i:\delta_\mu(b,i)=c\}$. For all $s>0$,
    $[M^s]_{b,c}$ is the number of paths of length $s$ from $b$ to $c$
    in $\mathcal{A}_{\mu,a}$. Since all states are final, the number
    $N_s$ of words of length $s$ accepted by $\mathcal{A}_{\mu,a}$ is
    obtained by summing up all the entries of $M^s$ in the row
    corresponding to $a$. Because $\mathcal{A}_{\mu,a}$ has a loop of
    label $0$ in $a$, the number of words of length $s$ accepted by
    $\mathcal{A}_{\mu,a}$ and starting with $0$ is equal to the number
    $N_{s-1}$ of words of length $s-1$ accepted by
    $\mathcal{A}_{\mu,a}$.  Consequently, the number of words of
    length $s$ in the directive language $L_{\mu,a}$ is exactly
    $N_s-N_{s-1}$. Of course, the matrix $M$ can also be related to the
    morphism $\mu$ and $M_{b,c}$ is also the number of occurrences of
    $c$ in $\mu(b)$. In particular, summing up all entries in the row
    of $M^s$ corresponding to $a$ gives $|\mu^s(a)|$. Therefore, the
    number of words of length $s$ in the directive language
    $L_{\mu,a}$ is $|\mu^s(a)|-|\mu^{s-1}(a)|$ and we get that
      \begin{gather}\label{ssi}
    |\rep_S(n)|=s \Leftrightarrow n\in\{|\mu^{s-1}(a)|,\ldots,|\mu^s(a)|-1\}.
    \end{gather}
    In particular, if $0<n<|\mu(a)|$, we have $|\rep_S(n)|=1$ and in
    this case $\rep_S(n)=n$.  Since we have $\rep_S(0)=\varepsilon$
    and $\mu(a)=au$, for some $u\in\Sigma^*$, we get
    $y_0=a=\delta_\mu(a,\rep_S(0))$.  Hence, by the definition of
    $\mathcal{A}_{\mu,a}$, we have that $y_n=\delta_\mu(a,\rep_S(n))$
    for $n<|\mu(a)|$. Now let $s>0$ and assume that
    $y_n=\delta_\mu(a,\rep_S(n))$ for all $n<|\mu^s(a)|$. 
    Let $|\mu^s(a)|\le  n<|\mu^{s+1}(a)|$. There
    exist a unique $|\mu^{s-1}(a)|\le m<|\mu^s(a)|$ such that
       $$
       \mu^{s+1}(a)=
           \underbrace{\mu^{s-1}(a)u y_m v}_{\mu^{s}(a)} 
           \mu(u) \underbrace{x y_n y}_{\mu(y_m)}\mu(v),
       $$
                for some words $u,v,x,z$.
                Therefore
                $y_n=(\mu(y_m))_i$ for some
                $i\in\{0,\ldots,|\mu(y_m)|-1\}.$ Then by the definition of
                $\mathcal{A}_{\mu,a}$, we have
                $$y_n=\delta_\mu(y_m,i)=\delta_\mu(\delta_\mu(a,\rep_S(m)),i)
                =\delta_\mu(a,\rep_S(m)i)$$
                and in view of condition (\ref{ssi}) and again by the definition of
                $\mathcal{A}_{\mu,a}$, we get
                $$\val_S(\rep_S(m)i)=|\mu^{s}(a)|+|\mu(y_{|\mu^{s-1}(a)|})|
                +\cdots+|\mu(y_{m-1})|+i=n.$$
                Hence, $\rep_S(n)=\rep_S(m)i$ and the result follows.
    \kom{Now let $n\ge|\mu(a)|$ and assume that
    $y_\ell=\delta_\mu(a,\rep_S(\ell))$ for all $0\le \ell<n.$ There
    exist unique $m<n$ and $t\ge1$ such that
                $$\mu^{\omega}(a)=\mu^{t-1}(a)y_{|\mu^{t-1}(a)|}\cdots y_m \cdots
                y_{|\mu^{t}(a)|-1} \mu(y_{|\mu^{t-1}(a)|})
                \cdots\underbrace{\cdots
                  y_n\cdots}_{\mu(y_m)}\cdots.$$
Otherwise stated, for some $u,v,x,z$, we have the following
factorizations:
$$\mu^t(a)=\mu^{t-1}(a)uy_mv,\ \mu^{t+1}(a)=\mu^t(a)\mu(u)\mu(y_m)\mu(v)\text{ and } \mu(y_m)=xy_nz.$$
Therefore
                $y_n=(\mu(y_m))_i$ for some
                $i\in\{0,\ldots,|\mu(y_m)|-1\}.$ Then by the definition of
                $\mathcal{A}_{\mu,a}$, we have
                $$y_n=\delta_\mu(y_m,i)=\delta_\mu(\delta_\mu(a,\rep_S(m)),i)
                =\delta_\mu(a,\rep_S(m)i)$$
                and in view of condition (\ref{ssi}) and again by the definition of
                $\mathcal{A}_{\mu,a}$, we get
                $$\val_S(\rep_S(m)i)=|\mu^{t}(a)|+|\mu(y_{|\mu^{t-1}(a)|})|
                +\cdots+|\mu(y_{m-1})|+i=n.$$
                Hence, $\rep_S(n)=\rep_S(m)i$ and the result follows.}
\end{proof}
The following lemma is simply another formulation of the previous
result.

\begin{lemma}\label{cor:dir}
    Let $\mu:\Sigma\to\Sigma^*$ be a morphism prolongable on
    $a\in\Sigma$ and let $\mu^\omega(a)=y_0y_1y_2\cdots$. Let $S$ be
    the abstract numeration
    system built on the directive language $L_{\mu,a}$ of
    $(\mu,a)$ with the ordered alphabet $(\{0,\ldots,r_\mu-1\},0<\cdots<r_\mu-1)$. 
    Let $n\ge 0$ and $\rep_S(n)=w_0\cdots w_\ell$, where $w_j$'s are letters.
    Define $z_0:=\mu(a)$ and for
    $j=0,\ldots,\ell-1$, set $z_{j+1}:=\mu((z_j)_{w_j})$.
    Then, $y_n=(z_\ell)_{w_\ell}$.
\end{lemma}

\begin{example}
    Continue Example~\ref{exa:run2}. The fixed point $Y_1$ of $\mu_1$
    start with $$abeebebe=y_0\cdots y_7$$ and $\rep_S(7)=1010$. From
    Lemma~\ref{lem:dir}, $y_7=e$ has been generated applying
    $\mu_1$ to the letter in the position $\val_S(101)=4$, i.e.,
    $y_4=b$. We have $y_7=(\mu_1(b))_0$. In turn, $y_4$ occurs in the
    image by $\mu_1$ of the letter in the position $\val_S(10)=2$, $y_2=e$
    and we have $y_4=(\mu_1(e))_1$. Now $y_2$ appears in the image of the
    letter in the position $\val_S(1)=1$ and we have $y_2=(\mu_1(b))_0$.
\end{example}
The following result is obvious.

\begin{lemma}\label{lem:shape}
    Let $x,y$ be two infinite (unidimensional) words and $\lambda$,
    $\mu$ be two morphisms such that there exist letters $a,b$ such
    that $x=\lambda^\omega(a)$ and $y=\mu^\omega(b)$. The languages
    $L_{\lambda,a}$ and $L_{\mu,b}$ are equal if and only if
    $\Shape_\lambda(x)=\Shape_\mu(y)$.
\end{lemma}

\begin{example}
    If one considers the morphism $\mu_2$ defined by $a\mapsto ac$,
    $c\mapsto e$, $e\mapsto eg$, $g\mapsto h$ and $h\mapsto hc$ (which
    is derived from the first column of the bidimensional morphism in
    Example~\ref{exa:run}), we have the DFA $\mathcal{A}_{\mu_2,a}$
    depicted in Figure~\ref{fig:dfacar2}.
\begin{figure}[htbp]
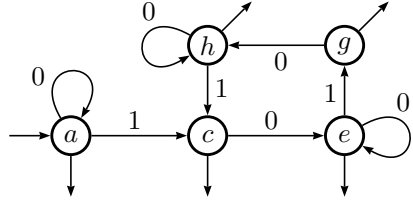

        \centering
\VCDraw{%
        \begin{VCPicture}{(0,-1)(6,3)}
% states
 \State[a]{(0,0)}{a}
 \State[c]{(3,0)}{c}
 \State[e]{(6,0)}{e}
 \State[h]{(3,2)}{h}
 \State[g]{(6,2)}{g}
% initial--final
\Initial{a}
\Final[s]{a}
\Final[s]{c}
\Final[s]{e}
\Final[ne]{h}
\Final[ne]{g}
% transitions
\LoopN{a}{0}
\EdgeL{a}{c}{1}
\EdgeL{c}{e}{0}
\LoopE{e}{0}
\EdgeL{e}{g}{1}
\EdgeL{g}{h}{0}
\EdgeL{h}{c}{1}
\LoopW{h}{0}
\end{VCPicture}
}
        \caption{The automaton $\mathcal{A}_{\mu_2,a}$.}
        \label{fig:dfacar2}
    \end{figure}
    The automata in Figure~\ref{fig:dfacar} and
    Figure~\ref{fig:dfacar2} clearly accept the same language (the
    first one being minimal).
\end{example}

Let $Y_2=(y_{1,n})_{n\ge 0}$ be the first column of $y$. This word
$Y_2$ is a unidimensional infinite word over a subset $\Sigma_2$ of
$\Sigma$.  It is clear that $Y_2$ is generated by a morphism $\mu_2$
derived from $\mu$. Since $y$ is shape-symmetric, thanks to
Remark~\ref{rem:shape} and to Lemma~\ref{lem:shape}, we have
$$L_{\mu_1,a}=L_{\mu_2,a}=:L.$$
We consider the abstract numeration system built upon this language
$L$ (with the natural ordering of digits). With all the above
discussion and in particular in view of Lemma~\ref{cor:dir},
it is clear that if $\rep_S(m)=ub$, $\rep_S(n)=vc$ where $b,c$ are letters,
then
\begin{equation}\label{eq:y}
(\mu(y_{\val_S(u),\val_S(v)}))_{b,c}=y_{m,n}.
\end{equation}

\begin{example}\label{exa:later}
    Consider the letter $c$ occurring in the position $(7,4)$ in
    the fixed point $y$ of $\mu$ underlined in Figure~\ref{fig:mu}. We
    have $(7,4)=(\val_S(1010),\val_S(101))$. If we consider the pair
    $(\val_S(101),\val_S(10))=(4,2)$, we get
    $(\mu(y_{4,2}))_{0,1}=(\mu(b))_{0,1}=c=y_{7,4}$.
    In other words, $y_{7,4}$ comes from $y_{4,2}$.
    We can continue this way.
    We have $b=y_{4,2}=(\mu(y_{2,1}))_{1,0}$ because
    $(\val_S(10),\val_S(1))=(2,1)$. Now $y_{2,1}=c=(\mu(y_{1,0}))_{0,1}$
    because $(\val_S(1),\val_S(\varepsilon))=(1,0)$. Finally
    $y_{1,0}=b=(\mu(y_{0,0}))_{1,0}=(\mu(a))_{1,0}$
    because $(\val_S(\varepsilon),\val_S(\varepsilon))=(0,0)$.
\end{example}

We now extend Definition~\ref{def:dir} to the multidimensional case.

\begin{definition}
For each $d$-dimensional morphism $\mu\colon \Sigma\to B_d(\Sigma)$ and for each letter $a\in\Sigma$, define a DFA $\mathcal{A}_{\mu,a}$ over the alphabet
$\{0,\ldots,r_\mu-1\}^d$ where $r_\mu=\max \{|\mu(b)|_i
:b\in\Sigma,i=1,\ldots, d\}$. The set of states is $\Sigma$, the initial state
is $a$ and all states are final. The (partial) transition function is defined by
$$\delta_\mu(b,\mathbf{n})=(\mu(b))_{\mathbf{n}},$$
for all $b\in\Sigma$ and $\mathbf{n}\le |\mu(b)|.$
\end{definition}

Thanks to \eqref{eq:y}, the automaton~$\mathcal{A}_{\mu,a}$ is such that,
for all $m,n\ge 0$,
$$y_{m,n}=\delta_{\mu}(a,(\rep_S(m),\rep_S(n))^{0}),$$
where we have padded the shortest word with enough $0$'s to make two
words of the same length as in Definition~\ref{diese}. If we consider
the coding $\nu$ as the output function, the
corresponding DFAO generates $x$ as an $S$-automatic sequence. 
Note that padding with $0$'s works correctly since $0$ 
is the lexicographically smallest letter and the directive language 
$L$ does not contain any words starting with $0$. This
concludes the first part. 

\begin{example}
                Consider the $2$-dimensional morphism $\mu$ of Example \ref{exa:run}
                and its fixed point $\mu^{\omega}(a)$ depicted in Figure \ref{fig:mu}.
                If $S=(L,\{0,1\},0<1)$ is the abstract numeration system constructed
                on $L=\{\varepsilon,1,10,100,101,1000,1001,1010,\ldots\},$
                then the corresponding DFAO depicted in Figure \ref{fig:dfao2},
                where the output function is the identity,
                generates $\mu^{\omega}(a)$ as an $S$-automatic word.
                For instance, if we continue Example \ref{exa:later},
                by reading              $(\rep_S(7),\rep_S(4))^{0}=(1010,0101)$, we get
                $$y_{0,0}=a\stackrel{(1,0)}{\rightarrow}y_{1,0}=b
                \stackrel{(0,1)}{\rightarrow}y_{2,1}=c
                \stackrel{(1,0)}{\rightarrow}y_{4,2}=b
                \stackrel{(0,1)}{\rightarrow}y_{7,4}=c, $$
                and the letters appearing in this sequence of transitions
                are exactly the underlined ones in Figure       \ref{fig:mu}.

\begin{figure}[htbp]
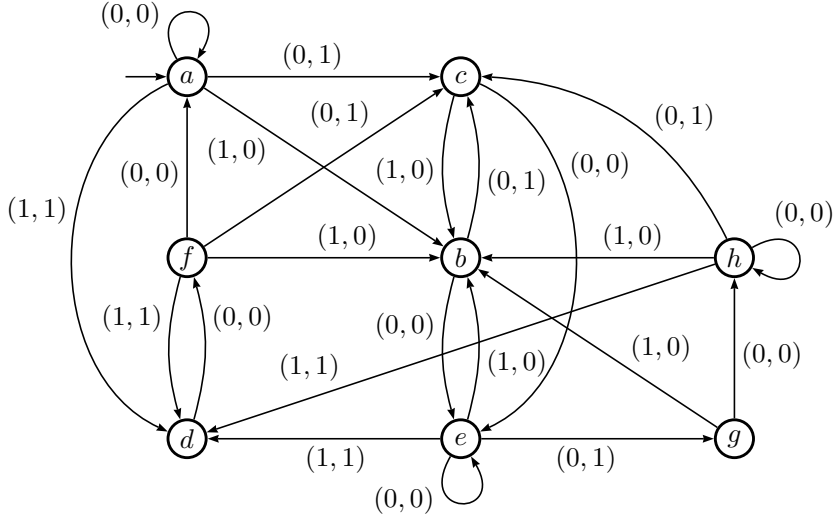

        \centering
\VCDraw{%
        \begin{VCPicture}{(0,-6)(12,6)}
% states
 \State[a]{(0,4)}{A}
 \State[b]{(6,0)}{B}
 \State[c]{(6,4)}{C}
 \State[d]{(0,-4)}{D}
 \State[e]{(6,-4)}{E}
 \State[f]{(0,0)}{F}
 \State[g]{(12,-4)}{G}
 \State[h]{(12,0)}{H}
% initial--final
\Initial[w]{A}
% transitions
\LoopN{A}{(0,0)}
\EdgeL{A}{C}{(0,1)}
\EdgeR[.25]{A}{B}{(1,0)}
\VArcR{arcangle=-65}{A}{D}{(1,1)}

\EdgeL{F}{A}{(0,0)}
\EdgeL[.68]{F}{C}{(0,1)}
\EdgeL[.6]{F}{B}{(1,0)}
\ArcR[.3]{F}{D}{(1,1)}

\ArcR[.35]{B}{E}{(0,0)}
\ArcR[.4]{B}{C}{(0,1)}

\ArcR[.53]{C}{B}{(1,0)}
\VArcL[.3]{arcangle=65}{C}{E}{(0,0)}

\ArcR[.7]{D}{F}{(0,0)}

\LoopS{E}{(0,0)}
\ArcR{E}{B}{(1,0)}
\EdgeL{E}{D}{(1,1)}
\EdgeR{E}{G}{(0,1)}

\EdgeR{G}{H}{(0,0)}
\EdgeR[.35]{G}{B}{(1,0)}

\LoopE{H}{(0,0)}
\EdgeR[.35]{H}{B}{(1,0)}
\EdgeR[.77]{H}{D}{(1,1)}
\VArcR{arcangle=-30}{H}{C}{(0,1)}

\end{VCPicture}
}
        \caption{DFAO generating $\mu^{\omega}(a)$ as an $S$-automatic word.}
        \label{fig:dfao2}
\end{figure}
\end{example}

%===================================

\medskip

{\bf Part 2.}
Assume that $x=(x_{m,n})_{m,n\ge 0}$ is a $2$-dimensional $S$-automatic
infinite word over $\Gamma$ for some abstract numeration system $S=(L,\Sigma,<)$
where $\varepsilon\in L$ and $\Sigma=\{a_1, \ldots,a_{r}\}$ with $a_1<\cdots<a_{r}$.
Let $\mathcal{A}=(Q_\mathcal{A},q_0,(\Sigma_{\#})^2,\delta_\mathcal{A},\Gamma,\tau_\mathcal{A})$
be a deterministic finite automaton with output generating~$x$ where we may assume that $\#=:a_0$ is a symbol not belonging to $\Sigma$ and that $a_0<a_1$.
Recall that this means that
$x_{m,n}=\tau_\mathcal{A}(\delta_\mathcal{A}(q_0,(\rep_S(m),\rep_S(n))^{\#}))$ for all $m,n\ge 0$.
Without loss of generality, we suppose that
$\delta_\mathcal{A}(q,(\#,\#))=q$, for all $q\in Q_\mathcal{A}$.
In this part we prove that $x$ can be represented as the image by a coding
of a morphic shape-symmetric $2$-dimensional infinite word. We do the proof
in three steps. First, we show that $x$ can be obtained applying an erasing map
to a fixed point of a uniform $2$-dimensional morphism. In the second step we
prove that $x$ is morphic. The generating morphism $\mu$ and the coding $\nu$ are obtained
using a construction represented for dimension one in~\cite{AllSha}.
Finally, we show that the considered fixed point of $\mu$ is shape-symmetric.

\begin{definition} \label{canonicalmorphism} Let $d\ge1$.  Any DFA of
    the form $\mathcal{A}=(Q,q_0,\Sigma^d,\delta,F)$, where
    $\Sigma=\{a_0,a_1, \ldots, a_{r}\}$ with the ordering
    $a_i<a_{i+1}$ for all $0\le i\le r-1$, can be canonically
    associated with a $d$-dimensional morphism denoted by
    $\mu_\mathcal{A}\colon Q \to B_d(Q)$ and defined as follows. The
    image of a letter $q\in Q$ is a $d$-dimensional square~$x$ of size
    $r+1$ defined by
    $x_{\mathbf{n}}=\delta(q,(a_{n_1},\cdots,a_{n_d})),$ for all $\mathbf{0}\le
    \mathbf{n}=(n_1,\ldots,n_d)\le r.\mathbf{1}$.
\end{definition}

\begin{example}
Consider the alphabet $\Sigma=\{\#,a,b\}$ with $\#<a<b$ and the automaton $\mathcal{A}$ depicted in Figure \ref{fig:dfao1}
with added loops of label $(\#,\#)$ on all states.
Then we get
$$\mu_\mathcal{A}(p)=\begin{array}{|c|c|c|c|}
\hline
                p & q & q \\
\hline
    p & p & s \\
\hline
    q & p & s \\
\hline
\end{array}\; ,
\mu_\mathcal{A}(q)=\begin{array}{|c|c|c|c|}
\hline
                q & p & q \\
\hline
    p & s & q \\
\hline
    p & q & s \\
\hline
\end{array}\; ,
\mu_\mathcal{A}(r)=\begin{array}{|c|c|c|c|}
\hline
                r & s & s \\
\hline
    p & r & s \\
\hline
    p & r & p \\
\hline
\end{array}\; ,
\mu_\mathcal{A}(s)=\begin{array}{|c|c|c|c|}
\hline
                s & r & s \\
\hline
    r & q & s \\
\hline
    r & s & r \\
\hline
\end{array}$$
and ${\mu_\mathcal{A}}^\omega(p)$ is the $2$-dimensional infinite word depicted
in Figure~\ref{tab:ex2}. Notice that ${\mu_\mathcal{A}}^\omega(p)$ is different from the $S$-automatic word given in Figure~\ref{tab:ex1}. However, by erasing some rows and columns in Figure~\ref{tab:ex2}, we obtain exactly the word in Figure~\ref{tab:ex1}.

\begin{figure}[htpb]
$$\begin{array}{|ccc|ccc|ccc|ccc}
\hline
p & q & q & q & p & q & q & p & q & q & p & \cdots \\
p & p & s & p & s & q & p & s & q & p & s & \\
q & p & s & p & q & s & p & q & s & p & q & \\
\hline
p & q & q & p & q & q & s & r & s & p & q & \\
p & p & s & p & p & s & r & q & s & p & p & \\
q & p & s & q & p & s & r & s & r & q & p & \\
\hline
q & p & q & p & q & q & s & r & s & p & q & \\
p & s & q & p & p & s & r & q & s & p & p & \\
p & q & s & q & p & s & r & s & r & q & p & \\
\hline
p & q & q & q & p & q & q & p & q & p & q & \\
p & p & s & p & s & q & p & s & q & p & p & \\
\vdots & & &  &   &   &   &   &   &   &   & \ddots \\
\end{array}$$
\caption{The fixed point ${\mu_\mathcal{A}}^\omega(p)$.}
        \label{tab:ex2}
    \end{figure}
\end{example}

By assumption, $L$ is a regular language over $\Sigma$. Hence, there
exists a DFA accepting $L$ and we may easily modify it to obtain a DFA
$\mathcal{L}=(Q_\mathcal{L}, \ell_0 ,\Sigma_{\#}, \delta_\mathcal{L}, F_\mathcal{L})$ accepting
$\{\#\}^*L$ and satisfying $\delta_\mathcal{L}(l_0,\#)=l_0$. Note that $l_0$
is a final state since $\varepsilon \in L$. Let us next define a
``product'' automaton $\mathcal{P}=(Q,p_0,(\Sigma_{\#})^2,\delta,F)$
imitating the behavior of~$\mathcal{A}$ and two copies of the
automaton~$\mathcal{L}$, one for each dimension. The set of states
of $\mathcal{P}$ is the Cartesian product $Q= Q_\mathcal{A} \times
Q_\mathcal{L} \times Q_\mathcal{L}$, where the initial state $p_0$ is
$(q_0,\ell_0,\ell_0)$. The transition function $\delta \colon Q
\times (\Sigma_{\#})^2 \to Q$ is defined by
\[
\delta((q,k,\ell),(a,b))= (\delta_\mathcal{A}(q,(a,b)), \delta_\mathcal{L}(k,a), \delta_\mathcal{L}(\ell,b)),
\]
where $(q,k,\ell)$ belongs to $Q$ and $(a,b)$ is a pair of letters in $(\Sigma_{\#})^2$.
The set of final states is $F=Q_\mathcal{A} \times F_\mathcal{L} \times F_\mathcal{L}$.
Let $y=(y_{m,n})_{m,n \geq 0}$ be the infinite word satisfying
\[
y_{m,n}=\delta(p_0,(\rep_S(m),\rep_S(n))^{\#}).
\]
Note that both the first and the second component of $(\rep_S(m),\rep_S(n))^{\#}$ belong
to the language $\{\#\}^*L$ and, therefore, $\delta(p_0,(\rep_S(m),\rep_S(n))^{\#})$ is a final state.
Define $\tau \colon F \to \Gamma$ to be the coding satisfying  $\tau((q,k,\ell))=\tau_\mathcal{A}(q)$
for all $(q,k,\ell) \in F$.
By construction, it is clear that $\tau(y)=(x_{m,n})_{m,n\ge 0}$.
\kom{\begin{example}
Consider again the automaton given by Figure \ref{fig:dfao1}. The word $x$ introduced
in this proof corresponds exactly to the $2$-dimensional infinite word depicted in
Figure \ref{tab:ex1}.
\end{example}}
We consider the canonically associated morphism
$\mu_\mathcal{P}\colon Q \to \mathcal{B}_2(Q)$ given in
Definition~\ref{canonicalmorphism}. Note that $\mu_\mathcal{P}$ is
prolongable on $p_0$, since
$\delta(p_0,(a_0,a_0))=(\delta_A(q_0,(\#,\#)),\delta_\mathcal{L}(l_0,\#),\delta_\mathcal{L}(l_0,\#))=(q_0,l_0,l_0)=p_0$.
Moreover, ${\mu_\mathcal{P}}^\omega(p_0)$ is shape-symmetric with
respect to $\mu_\mathcal{P}$, since $\mu_{\mathcal{P}}(q)$ is a
square of size $r+1$ for all $q \in Q$.

\begin{example}\label{exa:continued}
Let us continue Example~\ref{exa:S-aut} and consider again the abstract numeration system $S=(\{a,ba\}^*\{\varepsilon,b\},\{a,b\},a<b)$ and the DFAO depicted in Figure~\ref{fig:dfao1},
with additional loops of label $(\#,\#)$
on all states.
The minimal automaton of $\{\#\}^* \{a,ba\}^*\{\varepsilon,b\}$ is depicted in Figure~\ref{fig:minimal}.
\begin{figure}[htbp]
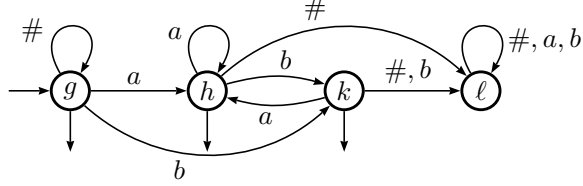

        \centering
\VCDraw{%
        \begin{VCPicture}{(0,0)(9,3)}
% states
 \State[g]{(0,1.5)}{G}
 \State[h]{(3,1.5)}{H}
 \State[k]{(6,1.5)}{K}
 \State[\ell]{(9,1.5)}{L}
% initial--final
\Initial[w]{G}
\Final[s]{G}
\Final[s]{H}
\Final[s]{K}
% transitions
\LoopN{G}{\#}
\EdgeL{G}{H}{a}
\VArcR{arcangle=-42}{G}{K}{b}

\VArcL{arcangle=42}{H}{L}{\#}
\LoopN{H}{a}
\ArcL[.6]{H}{K}{b}

\EdgeL{K}{L}{\#,b}
\ArcL[.6]{K}{H}{a}

\LoopN[.8]{L}{\#,a,b}
\end{VCPicture}
}
        \caption{The minimal automaton accepting $\{\#\}^* \{a,ba\}^*\{\varepsilon,b\}$.}
        \label{fig:minimal}
\end{figure}
If $\mathcal{P}$ is the corresponding product automaton, then the fixed point
${\mu_\mathcal{P}}^\omega((p,g,g))$ of $\mu_\mathcal{P}$
is the $2$-dimensional infinite word depicted in Figure~\ref{tab:productfixedpoint}.
 \begin{figure}[htpb]
$$\begin{array}{|ccc|ccc|ccc}
\hline
(p,g,g)    & (q,h,g)    & (q,k,g)    & (q,\ell,g)    & (p,h,g)    & (q,k,g)    & (q,\ell,g)    & (p,h,g)    & \cdots \\
(p,g,h)    & (p,h,h)    & (s,k,h)    & (p,\ell,h)    & (s,h,h)    & (q,k,h)    & (p,\ell,h)    & (s,h,h)    & \\
(q,g,k)    & (p,h,k)    & (s,k,k)    & (p,\ell,k)    & (q,h,k)    & (s,k,k)    & (p,\ell,k)    & (q,h,k)    & \\
\hline
(p,g,\ell) & (q,h,\ell) & (q,k,\ell) & (p,\ell,\ell) & (q,h,\ell) & (q,k,\ell) & (s,\ell,\ell) & (r,h,\ell) & \\
(p,g,h)    & (p,h,h)    & (s,k,h)    & (p,\ell,h)    & (p,h,h)    & (s,k,h)    & (r,\ell,h)    & (q,h,h)    & \\
(q,g,k)    & (p,h,k)    & (s,k,k)    & (q,\ell,k)    & (p,h,k)    & (s,k,k)    & (r,\ell,k)    & (s,h,k)    & \\
\hline
(q,g,\ell) & (p,h,\ell) & (q,k,\ell) & (p,\ell,\ell) & (q,h,\ell) & (q,k,\ell) & (s,\ell,\ell) & (r,h,\ell) & \\
(p,g,h)    & (s,h,h)    & (q,k,h)    & (p,\ell,h)    & (p,h,h)    & (s,k,h)    & (r,\ell,h)    & (q,h,h)    & \\
(p,g,\ell)    & (p,h,\ell)    & (s,k,\ell)    & (p,\ell,\ell)    & (p,h,\ell)    & (s,k,\ell)    & (r,\ell,\ell)    & (s,h,\ell)    & \\
\hline
(p,g,\ell) & (q,h,\ell) & (q,k,\ell) & (q,\ell,\ell) & (p,h,\ell) & (q,k,\ell) & (q,\ell,\ell) & (p,h,\ell) & \\
(p,g,\ell) & (p,h,\ell) & (s,k,\ell) & (p,\ell,\ell) & (s,h,\ell) & (q,k,\ell) & (p,\ell,\ell) & (s,h,\ell) & \\
\kom{(q,g,\ell) & (p,h,\ell) & (s,k,\ell) & (p,\ell,\ell) & (q,h,\ell) & (s,k,\ell) & (p,\ell,\ell) & (q,h,\ell) & \\}
\vdots     &            &            &               &            &            &               &           & \ddots \\
\end{array}$$
\caption{The fixed point ${\mu_\mathcal{P}}^\omega((p,g,g))$.}
        \label{tab:productfixedpoint}
    \end{figure}
\end{example}
Let $e$ be a new symbol. Recall that $\rho_e$ is the erasing map
given in Definition~\ref{def:rho}. Denote $\rho=\rho_e \circ
\lambda$, where $\lambda$ is a morphism on $Q \cup \{e\}$ defined by
\[
\lambda(p)= \left\{ \begin{array}{ll}
e & \text{if $p \not \in F$,}\\
p & \text{otherwise}.
\end{array}
\right.
\]
We claim that $y=\rho({\mu_\mathcal{P}}^\omega(p_0))$. Observe that
the infinite word $\lambda({\mu_\mathcal{P}}^\omega(p_0))$ is
$e$-erasable. Namely, all letters in a fixed column~$C$ of the
infinite bidimensional word ${\mu_\mathcal{P}}^\omega(p_0)$ are of
the form $(q,k,\ell)$ where the second component~$k$ is fixed. If
$k$ does not belong to $F_\mathcal{L}$, the word $\lambda(C)$ is a
unidimensional $e$-hyperplane of
$\lambda({\mu_\mathcal{P}}^\omega(p_0))$. Thus, the map~$\rho$
erases all columns where the second component~$k$ does not belong
to~$F_\mathcal{L}$. The same holds for rows and third components~$\ell$ of the
letters in~$Q$. Hence, the $2$-dimensional infinite word
$\rho({\mu_\mathcal{P}}^\omega(p_0))$ contains only letters
belonging to $F$. By the construction of the
morphism~$\mu_\mathcal{P}$, those letters are coming from the
automaton $\mathcal{P}$ by feeding it with words belonging
to~$((\Sigma_{\#})^2)^* \cap (\{\#\}^*L)^2$. More precisely, all
rows and columns not belonging to~$y$ are erased and
$(\rho({\mu_\mathcal{P}}^\omega(p_0)))_{m,n}$ is equal to
$\delta(p_0,(\rep_S(m),\rep_S(n))^{\#})=y_{m,n}$. Hence, defining
$\vartheta=\tau \circ \rho$, we get a map from $\Sigma$ to $\Gamma$
such that $x=\vartheta({\mu_\mathcal{P}}^\omega(p_0))$.

\begin{example}\label{exa:continued2}
We continue Example~\ref{exa:continued} and we consider this time
the bidimensional infinite $S$-automatic word depicted in Figure
\ref{tab:ex1}. This word is exactly the $2$-dimensional infinite
word obtained by first erasing all columns with $\ell$ as the second
component and all rows with $\ell$ as the third component from the
$2$-dimensional infinite word ${\mu_\mathcal{P}}^\omega((p,g,g))$
depicted in Figure~\ref{tab:productfixedpoint} and then mapping the
infinite word by~$\tau$.
\end{example}

Next we show that $x$ is morphic by getting rid of the erasing map~$\rho$.
We construct a morphism $\mu$ prolongable on some letter $\alpha$
and a coding $\nu$ such that $x=\nu(\mu^\omega(\alpha))$.
We follow the guidelines of \cite[Theorem 7.7.4]{AllSha}.
First we need the following definitions.

\begin{definition}
Let $\mu$ be a morphism on some finite alphabet $\Sigma$ and let $\Psi \subseteq \Sigma$. We say that a letter $a \in \Sigma$ is
\begin{enumerate}
\item[(i)] \emph{$(\mu, \Psi)$-dead} if the word $\mu^n(a) \in \Psi^*$ for every $n\geq0$.
\item[(ii)] \emph{$(\mu, \Psi)$-moribund} if there exists $m\geq0$ such that the word $\mu^m(a)$ contains at least one letter in $\Sigma \setminus \Psi$, and for every $n>m$, $\mu^n(a) \in \Psi^*$.
\item[(iii)] \emph{$(\mu, \Psi)$-robust} if there exist infinitely many  $n\geq0$ such that the word $\mu^n(a)$ contains at least one letter in $\Sigma \setminus \Psi$.
\end{enumerate}
\end{definition}
The following lemma from~\cite[Lemma 7.7.3]{AllSha} is valid also for multidimensional morphisms, since the proof is only based on the finiteness of the alphabet~$\Sigma$.

\begin{lemma}\label{lmAllSha}
Let $\mu$ be a morphism on some finite alphabet $\Sigma$ and let $\Psi \subseteq \Sigma$. Then there exists an integer $T\geq1$ such that the morphism $\varphi=\mu^T$ satisfies:
\begin{enumerate}
\item[(a)] If $a$ is $(\varphi, \Psi)$-moribund, then $\varphi^n(a) \in \Psi^*$ for all $n>0$ and $a \in \Sigma \setminus \Psi$.
\item[(b)] If $a$ is $(\varphi, \Psi)$-robust, then the word $\varphi^n(a)$ contains at least one letter in $\Sigma \setminus \Psi$ for all $n>0$.
\end{enumerate}
\end{lemma}

\begin{remark}\label{rem:psi}
Note that by Lemma~\ref{lmAllSha} a letter in $\Psi$ is either $(\varphi, \Psi)$-dead or $(\varphi, \Psi)$-robust and a letter in $\Sigma \setminus \Psi$ is either
$(\varphi, \Psi)$-moribund or $(\varphi, \Psi)$-robust.
\end{remark}

We may assume, by taking a power of $\mu_\mathcal{P}$ if necessary,
that $\mu_\mathcal{P}$ satisfies the properties (a) and (b) listed
for $\varphi$ in~Lemma~\ref{lmAllSha} with $\Psi=F^c:=Q\setminus F$. For the sake of simplicity, we use the words \emph{dead}, \emph{moribund} and \emph{robust} instead of $(\mu_\mathcal{P}, F^c)$-dead, $(\mu_\mathcal{P}, F^c)$-moribund and $(\mu_\mathcal{P}, F^c)$-robust from now on. 

Next we classify the states of $Q_\mathcal{L}$ and $Q$ into four categories.
The \emph{type of a state} $k \in Q_\mathcal{L}$ is
\[
T_k= \left \{ \begin{array}{ll}
\Delta & \text{if $k \not \in F_\mathcal{L}$ and $\delta_\mathcal{L}(k, a) \not \in F_\mathcal{L}$ for every $a \in \Sigma_{\#}$,}\\
M & \text{if $k \in F_\mathcal{L}$ and $\delta_\mathcal{L}(k, a) \not \in F_\mathcal{L}$ for every $a \in \Sigma_{\#}$,}\\
R_{F^c} & \text{if $k \not \in F_\mathcal{L}$ and there exists a letter $a \in \Sigma_{\#}$ such that $\delta_\mathcal{L}(k, a) \in F_\mathcal{L}$},\\
R_{F} & \text{if $k \in F_\mathcal{L}$ and there exists a letter $a \in \Sigma_{\#}$ such that $\delta_\mathcal{L}(k, a) \in F_\mathcal{L}$.}
\end{array}
\right.
\]

\noindent The \emph{type of a state} $p=(q,k,\ell) \in Q$ is
\[
T_p= \left \{ \begin{array}{ll}
\Delta & \text{if $p$ is dead,}\\
M & \text{if $p$ is moribund,}\\
R_{F^c} & \text{if $p\in F^c$ and $p$ is robust,}\\
R_{F} & \text{if $p\in F$ and $p$ is robust.}
\end{array}
\right.
\]
By these definitions, it is clear that the type of $(q,k,\ell) \in Q$
only depends on the types of~$k$ and $\ell\in Q_\mathcal{L}$ according to Figure~\ref{taType}.
Note that by the properties (a) and (b) of Lemma~\ref{lmAllSha}, %ajout (a), ok?
it suffices to consider transitions $\delta_\mathcal{L}(k,a)$ by each letter
$a \in \Sigma_{\#}$ instead of transitions $\delta_\mathcal{L}(k,w)$ by all
words $w$ in $(\Sigma_{\#})^*$. For instance, if the type of~$k$
is~$R_{F^c}$ and the type of $\ell$ is $R_{F}$, then $k \not \in
F_\mathcal{L}$ and $(q,k,\ell)$ belongs to~$F^c$. Moreover, there exist
$m,n\in[\![0,r]\!]$ such that $\delta_\mathcal{L}(k,a_m) \in F_\mathcal{L}$ and
$\delta_\mathcal{L}(\ell,a_n) \in F_\mathcal{L}$. This means that
$(\mu_\mathcal{P}((q,k,\ell)))_{m,n}$ belongs to $F$. Hence, by
Lemma~\ref{lmAllSha} and Remark~\ref{rem:psi}, $(q,k,\ell)$ is
robust.

\begin{figure}[htb!]
\begin{center}
\begin{tabular}{|c|cccc|}
\hline
\backslashbox{$T_\ell$}{$T_k$} & $\Delta$ & $M$ & $R_{F^c}$ & $R_{F}$\\
\hline
$\Delta$ & $\Delta$ & $\Delta$ & $\Delta$ & $\Delta$\\
$M$ & $\Delta$ & $M$  & $\Delta$ & $M$\\
$R_{F^c}$ & $\Delta$ & $\Delta$ & $R_{F^c}$ & $R_{F^c}$\\
$R_{F}$ & $\Delta$ & $M$ & $R_{F^c}$ & $R_{F}$\\
\hline
\end{tabular}
\end{center}
\caption{Type $T_p$ of a letter $p=(q,k,\ell) \in Q$.}
\label{taType}
\end{figure}
Let us define two morphisms $\lambda_\Delta$ and $\lambda_M$ on $Q\cup \{e\}$ in a similar way as $\lambda$ was defined above :

\begin{eqnarray*}
\lambda_\Delta(p)&=& \left\{ \begin{array}{ll}
e & \text{if $p$ is dead,}\\
p & \text{otherwise};
\end{array}
\right.\\
\lambda_M(p)&=& \left\{ \begin{array}{ll}
e & \text{if $p$ is moribund,}\\
p & \text{otherwise}.
\end{array}
\right.
\end{eqnarray*}
By the property~(b) of Lemma~\ref{lmAllSha}, we know that if $p$ is
robust, then $\mu_\mathcal{P}(p)$ contains
at least one letter in $F$ and since every dead letter must belong to~$F^c$, the word
$\lambda_\Delta(\mu_\mathcal{P}(p))$ contains at least one letter in
$F$. For any $\ell\in Q_{\mathcal{L}}$, let us define a sequence $(d_\ell(i))_{0\le i\le h_\ell}$ such that $d_\ell(0)=0$, $d_\ell(h_\ell)=r+1$ and for all $i\in[\![0,h_\ell-1]\!]$, $d_\ell(i)<d_\ell(i+1)$ and there exists exactly one index $n\in[\![d_\ell(i),d_\ell(i+1)-1]\!]$ satisfying 
\begin{equation} \label{h}
\delta_\mathcal{L}(\ell,a_n)\in F_{\mathcal{L}}.
\end{equation} 
Note that $h_\ell$ is the number of letters $a_n\in\Sigma_\#$ satisfying condition \eqref{h}. 
Hence, for each robust letter
$p=(q,k,\ell)$, we get $h_k, h_\ell\ge 1$ and we may define the factorization
\begin{equation*}
\lambda_\Delta(\mu_\mathcal{P}(p))=
\left. \begin{array}{|cccc|}
\hline
w_p(0,0) & w_p(1,0) & \cdots & w_p(h_{k}-1,0)\\
w_p(0,1) & w_p(1,1) & \cdots & w_p(h_{k}-1,1)\\
\vdots   & \vdots & \ddots & \vdots\\
w_p(0,h_{\ell}-1) & w_p(1,h_{\ell}-1) & \cdots & w_p(h_{k}-1,h_{\ell}-1)\\
\hline
\end{array}
\right. ,
\end{equation*}
where each bidimensional picture 
$$w_p(i,j)=\lambda_\Delta(\mu_\mathcal{P}(p))[(d_k(i),d_\ell(j)),(d_k(i+1)-1,d_\ell(j+1)-1)]$$
contains exactly one letter in $F$. Now we show that if $p$ is a
robust state, the bidimensional picture
$\lambda_M(\lambda_\Delta(\mu_\mathcal{P}(p)))$ is $e$-erasable. If
$v:=\lambda_M(\lambda_\Delta(\mu_\mathcal{P}(p)))$ is not
$e$-erasable, then there must exist $m,n\ge 0$ such that
$v_{m,n}=e$, $v_{m,n'} \neq e$ for some $n'$ and $v_{m',n}\neq e$
for some $m'$. By construction, the
letter~$p'=(\mu_\mathcal{P}(p))_{m,n}=(q,k,\ell)$ is mapped to $e$
either if $T_{p'} = \Delta$ or if $T_{p'}=M$. By the same reason,
the letters $v_{m,n'}=(q',k,\ell')$ and $v_{m',n}=(q'',k',\ell)$
must be robust. Thus, there exist letters
$a_{m''},\,a_{n''}\in \Sigma_{\#}$ such that $\delta_\mathcal{L}(k,a_{m''})
\in F$ and $\delta_\mathcal{L}(\ell,a_{n''}) \in F$. Hence, it follows that
$p'=(q,k,\ell)$ is  robust, since the letter
$(\mu_\mathcal{P}(p'))_{m'',n''}$ belongs to $F$, which is a
contradiction. Then for each robust letter
$p=(q,k,\ell)$, for each $i$ with $0 \leq i < h_{k}$ and for each
$j$ with $0 \leq j < h_{\ell}$, write
$$(\rho_e(\lambda_M(w_p(i,j))))_{m,n}=:v_{p,i,j}(m,n)$$
where $(m,n)<\mathbf{s}_{p,i,j}:=|\rho_e(\lambda_M(w_p(i,j)))|$.
Note that the picture $\lambda_M(w_p(i,j))$ is $e$-erasable as a factor
of the $e$-erasable picture $\lambda_M(\lambda_\Delta(\mu_\mathcal{P}(p)))$.
Now we are ready to introduce a $2$-dimensional morphism $\mu$ on a new
alphabet $\Xi$ and a coding $\nu' \colon \Xi \to Q$ such that $y=\nu'(\mu^\omega(\alpha))$ for a letter $\alpha \in \Xi$.
The alphabet of new symbols is
\[
\Xi = \{ \alpha(p,i,j) \mid \text{$p=(q,k,\ell)$ is robust,
$0 \leq i < h_{k}$ and $0 \leq j < h_{\ell}$} \}.
\]
We define the bidimensional pictures $u_{p,i,j}(m,n)$ for each
robust letter $p=(q,k,\ell)\in Q$, $(i,j)\in
[\![0,h_{k}-1]\!]\times [\![0,h_{\ell}-1]\!]$ and
$(m,n)\le\mathbf{s}_{p,i,j}$ as follows. If
$v_{p,i,j}(m,n)=(q',k',\ell')$, then $u_{p,i,j}(m,n)$ is a picture
of shape $(h_{k'},h_{\ell'})$ such that
\[
(u_{p,i,j}(m,n))_{i',j'}=\alpha(v_{p,i,j}(m,n),i',j')
\]
for $(i',j')\in[\![0,h_{k'}-1]\!]\times [\![0,h_{\ell'}-1]\!]$.
The image of $\alpha(p,i,j)$ by morphism $\mu \colon \Xi \to \mathcal{B}_2(\Xi)$ is defined as the word
\[
\left.
\begin{array}{|cccc|}
\hline
u_{p,i,j}(0,0) & u_{p,i,j}(1,0) & \cdots & u_{p,i,j}(s_1-1,0)\\
u_{p,i,j}(0,1) & u_{p,i,j}(1,1) & \cdots & u_{p,i,j}(s_1-1,1)\\
\vdots   & \vdots & \ddots & \vdots\\
u_{p,i,j}(0,s_2-1) & u_{p,i,j}(1,s_2-1) & \cdots & u_{p,i,j}(s_1-1,s_2-1)\\
\hline
\end{array}
\right. ,
\]
where $(s_1,s_2)=\mathbf{s}_{p,i,j}$. Note that the above
concatenation of the pictures $u_{p,i,j}(m,n)$ is well defined.
Since all letters occurring on a row of~$w_p(i,j)$ are of the form
$(q',k',\ell')$ where the third component $\ell'$ is fixed, it means
that also the letters $v_{p,i,j}(m,n)$ and $v_{p,i,j}(m',n)$
occurring on the same row of $\rho_e(\lambda_M(w_p(i,j)))$ have the
same third component~$\ell'$. Hence,
$|u_{p,i,j}(m,n)|_{\widehat{1}}=|u_{p,i,j}(m',n)|_{\widehat{1}}=h_{\ell'}$
and the words $u_{p,i,j}(m,n)$ and $u_{p,i,j}(m',n)$ can be
concatenated in the direction $1$. The same holds for
$u_{p,i,j}(m,n)$ and $u_{p,i,j}(m,n')$ in the direction $2$. The
coding $\nu' \colon \Xi \to Q$ is defined by
\begin{equation}\label{eq:nu}
\nu'(\alpha(p,i,j))=\rho(w_p(i,j))).
\end{equation}
Note that by the definition of $w_p(i,j)$, there is only one letter
belonging to $F$ and the picture $\lambda(w_p(i,j))$ is
$e$-erasable, since only one letter is different from~$e$. Following
the proof of \cite[Theorem 7.7.4]{AllSha}, we may prove by induction
that
\begin{equation}\label{eqMorph}
\nu' \circ \mu^n \left (\; \begin{array}{|cccc|} \hline
\alpha(p,0,0) & \alpha(p,1,0) & \cdots & \alpha(p,h_{k}-1,0)\\
%\hline
\alpha(p,0,1) & \alpha(p,1,1) & \cdots & \alpha(p,h_{k}-1,1)\\
%\hline
\vdots   & \vdots & \ddots & \vdots\\
%\hline
\alpha(p,0,h_{\ell}-1) & \alpha(p,1,h_{\ell}-1) & \cdots & \alpha(p,h_{k}-1,h_{\ell}-1)\\
\hline
\end{array}
\; \right) = \rho \circ \mu_\mathcal{P}^{n+1}(p)
\end{equation}
for all robust letters $p=(q,k,\ell)$
and for all $n \geq 0$.

Since $\mu_\mathcal{P}$ is prolongable on $p_0$ and
$x={\vartheta(\mu_{\mathcal{P}}}^\omega(p_0)$ is a $2$-dimensional
infinite word, $p_0$ must be a robust
letter. Therefore, we have
$(w_{p_0}(0,0))_{0,0}=v_{p_0,0,0}(0,0)=p_0$. Thus,
$(u_{p_0,0,0}(0,0))_{0,0}=\alpha(p_0,0,0)$ and, consequently, the
morphism $\mu$ is prolongable on $\alpha:=\alpha(p_0,0,0)$.
By~\eqref{eqMorph}, we have
\begin{align*}
\nu'(\mu^{n+1}(\alpha))&=\left[
\begin{array}{cc}
\nu'(\mu^n(u_{p_0,0,0}(0,0))) & U\\
V & W
\end{array}
\right ]\\
&= \left[
\begin{array}{cc}
\rho(\mu_\mathcal{P}^{n+1}(p_0))) & U\\
V & W
\end{array}
\right ], \\
\end{align*}
for all $n\ge 0$, where $U,\,V$ and $W$ are bidimensional pictures.
Since $\rho(\mu_\mathcal{P}^{n+1}(p_0))$ tends to $y$ as $n$ tends
to infinity, we have
\[
\nu'(\mu^\omega(\alpha))=\rho(\mu_\mathcal{P}^\omega(p_0))=y.
\]
Hence, defining the coding $\nu \colon \Xi \to \Gamma$ as $\nu=\tau \circ \nu'$ we obtain
\[
\nu(\mu^\omega(\alpha))=\tau(y)=x.
\]

\begin{example}
Let us continue Example~\ref{exa:continued2}. Recall that the product automaton $\mathcal{P}$ is produced from the automaton $\mathcal{A}$ depicted in Figure~\ref{fig:dfao1} and the automaton $\mathcal{L}$ depicted in Figure~\ref{fig:minimal}.
Note that the type of the state~$\ell$ in $\mathcal{L}$ is $T_\ell=\Delta$ and all other states have type $R_F$. By Figure~\ref{tab:productfixedpoint}, we see that
\[
\mu_\mathcal{P}(p,g,g) = \left. \begin{array}{|ccc|}
\hline
(p,g,g) & (q,h,g) & (q,k,g)\\
(p,g,h) & (p,h,h) & (s,k,h)\\
(q,g,k) & (p,h,k) & (s,k,k)\\
\hline
\end{array}\right.
\]
and
\[
\mu_\mathcal{P}(q,h,g) = \left. \begin{array}{|ccc|}
\hline
(q,\ell,g) & (p,h,g) & (q,k,g)\\
(p,\ell,h) & (s,h,h) & (q,k,h)\\
(p,\ell,k) & (q,h,k) & (s,k,k)\\
\hline
\end{array} \right.. 
\]
Since $h_\ell$ is the number of letters $a_n \in \Sigma_\#$ such that $\delta_\mathcal{L}(\ell,a_n)\in F_{\mathcal{L}}$, we notice that $h_g=3$ and $h_h=2$.
By Figure~\ref{taType}, we have  
$\rho_e(\lambda_\Delta(\mu_\mathcal{P}(p,g,g)))=\rho_e(\mu_\mathcal{P}(p,g,g))=\mu_\mathcal{P}(p,g,g)$ 
and
\[
\rho_e(\lambda_\Delta(\mu_\mathcal{P}(q,h,g))) = \rho_e \left (\; \begin{array}{|ccc|}
\hline
e & (p,h,g) & (q,k,g)\\
e & (s,h,h) & (q,k,h)\\
e & (q,h,k) & (s,k,k)\\
\hline
\end{array} \; \right)=
\left. \begin{array}{|cc|}
\hline
 (p,h,g) & (q,k,g)\\
 (s,h,h) & (q,k,h)\\
 (q,h,k) & (s,k,k)\\
\hline
\end{array}\right..
\]  
Since all letters in $\lambda_\Delta(\mu_\mathcal{P}(p,g,g))=\mu_\mathcal{P}(p,g,g)$ belong to $F$, the picture $w_{(p,g,g)}(i,j)$ is a square of size 1 for $(i,j)\in[\![0,h_g-1]\!]\times [\![0,h_g-1]\!]$. Consequently, 
\[
\mathbf{s}_{(p,g,g),i,j}=|\rho_e(\lambda_M(w_{(p,g,g)}(i,j)))|=(1,1)
\]
and
\[
v_{(p,g,g),i,j}(0,0)=w_{(p,g,g)}(i,j)=(\mu_P(p,g,g))_{i,j}
\]
for $(i,j) \in[\![0,2]\!]\times [\![0,2]\!]$. Especially, we have $v_{(p,g,g),0,0}(0,0)=(p,g,g)$ and 
$v_{(p,g,g),1,0}(0,0)=(q,h,g)$.
Hence, $u_{(p,g,g),0,0}(0,0)$ is a picture of shape $(h_g,h_g)=(3,3)$ such that
\[
(u_{(p,g,g),0,0}(0,0))_{i',j'}=\alpha(v_{(p,g,g),0,0}(0,0),i',j')=\alpha((p,g,g),i',j')
\]
for $(i',j')\in[\![0,2]\!]\times [\![0,2]\!]$
and the image $\mu(\alpha((p,g,g),0,0))=u_{(p,g,g),0,0}(0,0)$ is
\[
\left. \begin{array}{|ccc|}
\hline
\alpha((p,g,g),0,0) & \alpha((p,g,g),1,0) & \alpha((p,g,g),2,0)\\ 
\alpha((p,g,g),0,1) & \alpha((p,g,g),1,1) & \alpha((p,g,g),2,1)\\
\alpha((p,g,g),0,2) & \alpha((p,g,g),1,2) & \alpha((p,g,g),2,2)\\
\hline
\end{array} \right..
\]
Similarly, $|u_{(p,g,g),1,0}(0,0)|=(h_h,h_g)=(2,3)$ and 
\[
(u_{(p,g,g),1,0}(0,0))_{i',j'}=\alpha(v_{(p,g,g),1,0)}(0,0),i',j')=\alpha((q,h,g),i',j')
\]
for $(i',j')\in[\![0,1]\!]\times [\![0,2]\!]$.
Thus, the image $\mu(\alpha((p,g,g),1,0))=u_{(p,g,g),1,0}(0,0)$ is
\[
\left. \begin{array}{|cc|}
\hline
\alpha((q,h,g),0,0) & \alpha((q,h,g),1,0) \\ 
\alpha((q,h,g),0,1) & \alpha((q,h,g),1,1) \\
\alpha((q,h,g),0,2) & \alpha((q,h,g),1,2) \\
\hline
\end{array} \right..
\]
Next we apply the coding $\nu$ to the images above.
Note that 
\[
\begin{array}{rclcrcl} \vspace{0.5mm}
w_{(q,h,g)}(0,0) &=& \left. \begin{array}{|cc|}
\hline
e & (p,h,g)\\
\hline
\end{array}\right., & & w_{(q,h,g)}(1,0) &=& \left. \begin{array}{|c|}
\hline
(q,k,g)\\
\hline
\end{array}\right.,\\ \vspace{0.5mm}
w_{(q,h,g)}(0,1) &=& \left. \begin{array}{|cc|}
\hline
e & (s,h,h)\\
\hline
\end{array}\right., & & w_{(q,h,g)}(1,1) &=& \left. \begin{array}{|c|}
\hline
(q,k,h)\\
\hline
\end{array}\right.,\\
w_{(q,h,g)}(0,2) &=& \left. \begin{array}{|cc|}
\hline
e & (q,h,k)\\
\hline
\end{array}\right., & & w_{(q,h,g)}(1,2) &=& \left. \begin{array}{|c|}
\hline
(s,k,h)\\
\hline
\end{array}\right..\\
\end{array}
\]
Hence, by~\eqref{eq:nu}, we have
$\nu'(\mu(\alpha((p,g,g),0,0))) =\mu_\mathcal{P}(p,g,g)$
and
\[
\nu'(\mu(\alpha((p,g,g),1,0)))=
\begin{array}{|cc|}
\hline
 (p,h,g) & (q,k,g)\\
 (s,h,h) & (q,k,h)\\
 (q,h,k) & (s,k,k)\\
\hline
\end{array}.
\]
Since $\nu=\tau \circ \nu'$, the infinite word
$\nu(\mu^\omega(\alpha((p,g,g),0,0)))$ begins with 
\[
\nu \left( \mu(\alpha((p,g,g),0,0)) \odot^1 \mu(\alpha((p,g,g),1,0)) \right)=
\left. \begin{array}{|ccccc|}
\hline
p&q&q&p&q\\
p&p&s&s&q\\
q&p&s&q&s\\
\hline
\end{array} \right. ,
\]
which is exactly the left upper corner of the infinite word depicted in Figure~\ref{tab:ex1}. 

\end{example}

Finally, we have to show that $w=\mu^\omega(\alpha)$ is
shape-symmetric, that is \kom{In other words, we have to prove that
} for all $m,n\ge 0$, if $|\mu(w_{m,n})|=(s,t)$ then
$|\mu(w_{n,m})|=(t,s)$. First, observe that if $p=(q,k,\ell)$ is a
robust letter of $Q$, $0\le i < h_k$ and
$0\le j < h_\ell$, then the shape of $\mu(\alpha(p,i,j))$ does not
depend on $q$. More precisely, we have
\begin{equation}\label{eqshape}
|\mu(\alpha(p,i,j))|=\left( \sum_{m=0}^{s_1-1} |u_{p,i,j}(m,0)|_1, \sum_{n=0}^{s_2-1} |u_{p,i,j}(0,n)|_2 \right),
\end{equation}
where $(s_1,s_2)=\mathbf{s}_{p,i,j}$ does not depend on $q$,  the component $|u_{p,i,j}(m,0)|_1$ does not depend on $q$, $\ell$ and $j$
and, similarly, $|u_{p,i,j}(0,m)|_2$ does not depend on $q$, $k$ and
$i$. Moreover, for all $d\ge 0$, we have $|\mu^d(\alpha)|=(t_d,t_d)$
for some integer $t_d \ge0$, since $\alpha=\alpha(p_0,0,0)$ where
the second and the third component of $p_0=(q_0,l_0,l_0)$ are equal.
Hence, it suffices to show for all $m,n\ge 0$ that if
$w_{m,n}=\alpha((q,k,\ell),i,j)$ then
$w_{n,m}=\alpha((q',\ell,k),j,i)$ for some $q'$ in $Q_\mathcal{A}$.
We prove this by induction on the power~$d$ of~$\mu$. Assume that
for all $m,n \in [\![0,t_d-1]\!]$, if $(\mu^d(\alpha))_{m,n}=
\alpha((q,k,\ell),i,j)$ then
$(\mu^d(\alpha))_{n,m}=\alpha((q',\ell,k),j,i)$ for some $q'\in
Q_\mathcal{A}$. For $d=0$, the assumptions are clearly satisfied.
Consider now the letter
$$w_{m,n}=(\mu^{d+1}(\alpha))_{m,n}=:\alpha((q,k,\ell),i,j),$$ where
$m,n\in[\![0,t_{d+1}-1]\!]$ and $m$ or $n$ belongs to
$[\![t_d,t_{d+1}-1]\!]$. There exist unique
$m',n'\in[\![0,t_d-1]\!]$ such that $w_{m,n}$ is generated by
applying $\mu$ to
$$w_{m',n'}=(\mu^{d}(\alpha))_{m',n'}=:\alpha((q',k',\ell'),i',j').$$
By definition of $\mu$, there exists a unique pair $(m'',n'')<
\mathbf{s}_{(q',k',l'),i',j'}$ such that
$$(u_{(q',k',\ell'),i',j'}(m'',n''))_{i,j}=
\alpha(v_{(q',k',\ell'),i',j'}(m'',n''),i,j)= w_{m,n}.$$ By
induction hypothesis, we can write
$$w_{n',m'}=(\mu^{d}(\alpha))_{n',m'}=\alpha((q'',\ell',k'),j',i'),$$
where $q''\in Q_\mathcal{A}$ and by (\ref{eqshape}) we have
$$(|\mu(w_{n',m'})|_1,|\mu(w_{n',m'})|_2)=(|\mu(w_{m',n'})|_2,|\mu(w_{m',n'})|_1).$$
Therefore $w_{n,m}$ must be generated by applying $\mu$ to
$w_{n',m'}$. Moreover
\begin{gather*}
(|u_{(q'',\ell',k'),j',i'}(n'',m'')|_1,|u_{(q'',\ell',k'),j',i'}(n'',m'')|_2)\hspace{2cm}\\ =(|u_{(q',k',\ell'),i',j'}(m'',n'')|_2,|u_{(q',k',\ell'),i',j'}(m'',n'')|_1).
\end{gather*}
Thus, we conclude that
$$(u_{(q'',\ell',k'),j',i'}(n'',m''))_{j,i}=\alpha(v_{(q'',\ell',k'),j',i'}(n'',m''), j,i)=w_{n,m}.$$
Therefore we get that $v_{(q'',\ell',k'),j',i'}(n'',m'')=(q''',\ell,k)$
for some $q''' \in Q_\mathcal{A}$. Hence,
$$w_{n,m}=(\mu^{d+1}(\alpha))_{n,m}=\alpha((q''',\ell,k),j,i)$$
and the result follows.

\end{document}